\documentclass[pdflatex,sn-mathphys-ay]{sn-jnl}

\usepackage{amsmath, amssymb, amsthm}
\usepackage{natbib}
\usepackage{float}
\usepackage{graphicx}
\usepackage{booktabs}
\usepackage{multirow}
\usepackage{algorithm}
\usepackage{algorithmicx}
\usepackage{algpseudocode}
\usepackage{listings}
\usepackage{xcolor}
\usepackage{tikz}
\usepackage{pgfplots}
\pgfplotsset{compat=1.18}
\pgfplotsset{compat=newest}
\usepgfplotslibrary{groupplots}
\usetikzlibrary{positioning}
\usepackage{tabularx}
\usepackage{enumitem}
\usepackage{hyperref}
\usepackage{cleveref}

\usepackage{csquotes}
\usepackage{microtype}
\usepackage{placeins}
\usepackage{xurl}
\definecolor{c01}{HTML}{0C9756}
\definecolor{c02}{HTML}{30602D}
\definecolor{c03}{HTML}{096A9A}
\definecolor{c04}{HTML}{0B4199}
\definecolor{c05}{HTML}{F8850F}
\definecolor{c06}{HTML}{E75E2E}
\definecolor{c07}{HTML}{CB4149}
\definecolor{c08}{HTML}{A92E5E}
\definecolor{c09}{HTML}{84206B}
\definecolor{c10}{HTML}{5D126E}
\definecolor{c11}{HTML}{360961}
\definecolor{c12}{HTML}{11092E}

\theoremstyle{definition}
\newtheorem{definition}{Definition}[section]
\newtheorem{remark}{Remark}[section]

\theoremstyle{plain}
\newtheorem{theorem}{Theorem}[section]
\newtheorem{proposition}{Proposition}[section]

\newcommand{\QEP}{\mathcal{P}}
\newcommand{\QEPF}{\mathcal{P}(u)}
\newcommand{\Vitality}{\mathrm{V}}
\newcommand{\MRQ}[1]{\mathrm{M}(#1)}     

\title[Quantile-Based Effectiveness Persistence]{Quantile-Based Effectiveness Persistence Function: A Tail-Focused Metric with Theory, Estimation, and Application to Biosimilar Evaluation}

\author[1]{\fnm{Sankaran} \sur{P.G.}}
\author*[1]{\fnm{Prasanth} \sur{V.P.}}\email{prasanth.stat@gmail.com}
\author[2]{\fnm{Midhu} \sur{N.N.}}

\affil[1]{\orgdiv{Department of Statistics}, \orgname{Cochin University of Science and Technology}, \orgaddress{\city{Kochi}, \country{India}}}
\affil[2]{\orgdiv{Biostatistics}, \orgname{IQVIA}, \orgaddress{\city{Kochi}, \country{India}}}

\abstract{
In clinical studies, persistence, which measures the duration of time a patient continues to take a prescribed medication without discontinuation, is increasingly recognized as a critical indicator of adherence to medication. Adherence encompasses not only whether a patient takes their medication as prescribed but also the consistency and duration with which they do so. Among the various metrics used to evaluate adherence, persistence stands out as a particularly robust measure because it provides a temporal dimension, reflecting the sustained commitment of patients to their therapeutic regimens. This focus on persistence offers unique insights into adherence-related quality and performance, shedding light on the challenges and opportunities to optimize long-term medication use.   

The comparison of upper-tail clinical performance, which measures the extent to which very large responses persist among top responders, is often more decisive in therapy
evaluation than conventional summaries. In this paper, we introduce the quantile-based effectiveness persistence function defined  as the ratio between the tail mean  and the quantile function. The notion parallels expected shortfall in risk theory and is tailored to detect clinically meaningful deviations in the upper tail. We establish key properties and show that the function is equivalent to the first L-moment of the scaled tail, yielding robust inference tools. We derive a simple nonparametric estimator of the function and develop a bootstrap-calibrated two-sample (upper-tail) equivalence test. Simulation studies and real-data analysis illustrate that the proposed measures captures clinically relevant tail persistence that complements median and mean-based summaries.
}


\begin{document}
\maketitle

\section{Introduction}
\label{sec:intro}

In clinical research, biosimilars represent a growing area and face significant challenges, especially in emerging markets. The goal of a biosimilar development program is not to independently establish that the proposed biosimilar product is safe and effective, but rather to demonstrate that the proposed biosimilar product is highly similar to the reference biologic product. Biosimilar trials employ several statistical tools to demonstrate clinical similarity under an equivalence or non-inferiority framework \citep{ICH_E9R1_2019}.
In the equivalence framework, two-sided confidence intervals for differences or ratios between means of two populations are generally used for comparison \citep{Wellek2010}. In non-inferiority trials, one-sided confidence intervals for the difference of means, excluding the non-inferiority margin, are often used as supporting evidence for decision-making \citep{FDA2016}. When binary responses are collected, the risk ratio, risk difference, or odds ratio is compared to an equivalence margin. For pharmacokinetic (PK) bioequivalence, geometric mean ratios with 90\% confidence intervals within the 80-125\% range are required \citep{Schuirmann1987, ChowLiu2013}.
\medskip

The analysis of time-to-event data is commonly performed using the well-known concept of the hazard ratio. In the presence of covariates, landmark models are more useful when the proportional hazards assumption is doubtful \citep{Uno2014}. Comparison of clinical performance in the upper-tail, which measures the extent to which very large responses persist among the highest responders, is often more decisive in therapy evaluation than conventional summaries.
Existing statistical concepts mentioned above rarely quantify persistence among extreme responders. Traditional tools-such as means, hazard ratios, response rates, and survival times often assume proportional hazards or rely on arbitrary thresholds, and they can miss durable benefits among top responders. Several approaches have been proposed to describe upper-tail behavior in survival analysis, including mean residual life (MRL), tail mean or expected shortfall, Bonferroni and Lorenz curves, total time on test (TTT) transforms, and fixed-threshold survival metrics. MRL and tail mean are scale-dependent and unstable under heavy censoring, with limited clinical interpretability \citep{KleinMoeschberger2003, GuptaBradley2003,Emura2021}. Expected shortfall requires arbitrary tail-level choices and is sensitive to extremes \citep{Embrechts1997}. Bonferroni and Lorenz curves offer unique perspectives but lack direct clinical meaning and can be influenced by outliers \citep{Arnold2015}. TTT transforms serve as diagnostic tools in reliability theory rather than treatment-effect measures \citep{Bergman1981, KleinMoeschberger2003}.
\medskip

A probability distribution of a random variable can be represented either by a distribution function or by a quantile function. There are certain properties of the quantile function that are not shared by the distribution function \citep{NairSankaranDileepkumar2022}. In this paper, we introduce and study the quantile-based effectiveness-persistence function (QEPF) using the quantile function. Quantile-based methods have gained traction in survival analysis for their robustness and interpretability \citep{Peng2021, Kayid2024}. Recent advances in bootstrap and simultaneous inference for quantile processes \citep{Chernozhukov2016, Chernozhukov2022, Zhou2024, Hesterberg2015} enable rigorous uncertainty quantification for tail-focused measures like QEPF. The role of the quantile function and concepts derived from it is well studied in exploratory data analysis \citep{Parzen1979}. With heavy-tailed probability models, a single long-term survivor can have a marked effect on many survival measures based on the distribution function. In such cases, quantile-based measures are useful because they are less influenced by extreme observations. Furthermore, they provide direct analysis with limited information. For various properties and applications of quantile functions, one could refer to \citep{Gilchrist2000} and \citep{Nair2013}.
\medskip

The text is organized as follows: Sections \ref{sec:def} and \ref{sec:prop} introduce the quantile-based effectiveness persistence function and study its properties. The relationship with other survival measures is also discussed. Section \ref{sec:beta} details the persistence and stationary points specific to Beta distribution. A non-parametric estimator of the function is derived in Section \ref{sec:empirical-qepf}. Section \ref{sec:2sample_eq_test-qepf} propose a two-sample equivalence test based on QEPF. Simulation studies were conducted to assess the performance of the estimator and  equivalence test, which is presented in Section \ref{sec:simstudy}. The proposed equivalence test is applied to real-life data in Section \ref{sec:realdata}. Finally, the section \ref{sec:discussion-conclusion} summarizes the major findings of the work.

\section{Definition of Quantile-Based Effectiveness Persistence Function}
\label{sec:def}
Let $X$ be a continuous positive random variable with distribution function $F(x)$, and quantile function, $Q(u)=\inf\{x:F(x)\ge u\}$, $u\in(0,1)$. The vitality function of $X$ is defined as $v(x) \;=\; \mathbb{E}\!\left[X \,\middle|\, X > x\right]$ which measure mean of the random variable $X$ truncated at $x$. Then the \textit{quantile-based vitality function} is given by \citep{Nair2013}:
\[
\Vitality(u)\equiv \frac{1}{1-u}\int_u^1 Q(p)\,dp,\qquad u\in(0,1),
\]
which is a quantile-integrated tail mean, closely related to Expected Shortfall in risk theory \citep{AcerbiTasche2002,RockafellarUryasev2002}.

In finance, for a loss variable \(L\) and confidence level \(\alpha\in(0,1)\), the ratio \(\mathrm{ES}_\alpha(L)/\mathrm{VaR}_\alpha(L)\)---where \(\mathrm{VaR}_\alpha(L)=\inf\{x\in\mathbb{R}:\Pr(L\le x)\ge \alpha\}\) and (for continuous \(L\)) \(\mathrm{ES}_\alpha(L)=\mathbb{E}[L\,|\,L\ge \mathrm{VaR}_\alpha(L)]\)---summarizes tail severity relative to its entry threshold \citep{McNeil2005}.

We adapt this concept to a \emph{benefit} context by defining the Quantile-Based Effectiveness Persistence Function, $\QEPF$ as follows:

\begin{definition}
\label{def:QEPF}
For $u\in(0,1)$,
\begin{equation}
\label{eq:defQEPF}
\QEPF \equiv \frac{\Vitality(u)}{Q(u)}=\frac{1}{1-u}\int_u^1 \frac{Q(p)}{Q(u)}\,dp.
\end{equation}
\end{definition}

$\QEPF$ is interpreted as the \emph{average multiplicative uplift} among the top $(1-u)$ fraction relative to the entry threshold $Q(u)$, and since $Q(p)$ is strictly increasing for a continuous distribution, it follows that $\QEPF>1$ always. Moreover, $\lim_{u\uparrow 1}\QEPF=1$. $\QEPF$ for common continuous distributions are given in Table \ref{tab:qepf-summary}, and Figure \ref{fig:qepf-panels} shows the plot of $\QEPF$ for selected distributions.

\section{Properties of \texorpdfstring{$\QEPF$}{QEPF}}
\label{sec:prop}
\noindent The following properties are directly observed for $\QEPF$:

\begin{enumerate}
    \item \textbf{Scale invariance:} For $Y=aX$ with $a>0$, the measure remains unchanged under scaling, i.e., $\mathcal{P}_Y(u)=\mathcal{P}_X(u)$.

    \item \textbf{Location sensitivity:} For a shift $Y=X+b$ with $b>0$, the function is defined by $\mathcal{P}_Y(u)=\frac{V(u)+b}{Q(u)+b}$. As the shift increases, the measure attenuates toward unity, such that $\lim_{b \to \infty} \mathcal{P}_Y(u) = 1$.
    
    \item \textbf{Lower bound:} By construction, $\QEPF>1$ for all $u\in(0,1)$.

    \item \textbf{Relationship with Mean Residual Quantile:} The mean residual quantile function which is quantile version to mean residual life is given by $M(u) \;=\; \frac{1}{1-u}\int_{u}^{1}\!\bigl(Q(p)-Q(u)\bigr)\,dp,$. 0<u<1. \citep{Nair2013}. Then, we can show that,
    \begin{equation}\label{eq:pID1}
    \QEPF = 1+\frac{M(u)}{Q(u)}
    \end{equation}
    
\begin{proof}
\label{proof:mrq}
We have, 
$$\mathrm{M}(u)=\frac{1}{1-u}\int_u^1 \big(Q(p)-Q(u)\big)\,dp$$
then, $$\frac{M(u)}{Q(u)} = \frac{1}{(1-u)Q(u)}\int_u^1 \big(Q(p)-Q(u)\big)\,dp$$
$$\frac{M(u)}{Q(u)} = \frac{1}{(1-u)}\int_u^1\frac{ \big(Q(p)-Q(u)\big)}{Q(u)}\,dp = \QEPF-1$$
\end{proof}
    
    \item \textbf{Pareto characterization:} $\QEPF$ is constant in $u$ if and only if ($\iff$) the underlying distribution is Pareto I with quantile function $Q(u)=\sigma(1-u)^{-1/\alpha}$, $\alpha>1$, where $\sigma > 0 \text{ and } \alpha > 1$ \citep{Nair2013, MarshallOlkin2007}. In this case,
    \[
    \QEPF = \frac{\alpha}{\alpha-1}.
    \]
    For $\alpha\le 1$, the integral diverges, reflecting infinite mean and divergent $\QEPF$.
    
\begin{proof}
For $Q(u) = \sigma (1-u)^{-1/\alpha}$ with $\alpha>1$, the ratio $Q(p)/Q(u)=\{(1-p)/(1-u)\}^{-1/\alpha}$ yields
\[
\QEPF = \int_0^1 t^{-1/\alpha} dt = \frac{\alpha}{\alpha - 1}.
\]
Conversely, constancy implies a functional equation for $Q$ whose solution is $Q(u)=c(1-u)^{-1/\alpha}$.
\end{proof}

    %
\end{enumerate}

\begin{remark}
When $\QEPF$ remains constant (as in Pareto I with $\alpha>1$), it indicates that the tail of the distribution grows in a self-similar way. In practical terms, this means that the relative increase in extreme values stays consistent across higher quantiles. This pattern is particularly relevant in clinical trials for diseases or endpoints with heavy-tailed survival distributions, such as certain oncology or rare disease settings where a small subset of patients experiences very long survival. This property highlights that average equivalence does not guarantee tail equivalence, which is a critical insight for biosimilarity and risk assessment.
\end{remark}


\begin{table}[ht]
\setlength{\extrarowheight}{9pt}
\caption{Quantile function $ Q(u)$ and persistence function $\QEPF$ across continuous distributions.}
\label{tab:qepf-summary}
\centering
\raggedright
\begin{tabularx}{\textwidth}{>{\raggedright\arraybackslash}p{3cm} >{\raggedright\arraybackslash}p{3.9cm} X}
\toprule
\textbf{Distribution} & \textbf{$Q(u)$} & \textbf{$\QEPF$} \\
\midrule
Uniform $(0,1)$ & $u$ & $\dfrac{1+u}{2u}$ \\
Exponential $(\lambda)$ & $-\dfrac{1}{\lambda}\log(1-u)$ & $1 - \dfrac{1}{\log(1-u)}$ \\
Log-Logistic $(\alpha,\beta)$ \citep{Lawless2003,JohnsonKotz1994}
& $\alpha\!\left(\dfrac{u}{1-u}\right)^{1/\beta}$
&
$ \begin{aligned}[t]
&u^{-1/\beta}(1-u)^{\frac{1}{\beta}-1}\, B\!\Big(1+\tfrac{1}{\beta},\,1-\tfrac{1}{\beta}\Big)\, \\
&\Big[1-I_u\!\Big(1+\tfrac{1}{\beta},\,1-\tfrac{1}{\beta}\Big)\Big]
\end{aligned} $
\\
Power $(\alpha,\beta)$ \citep{NadarajahKotz2006}
& $\alpha\,u^{1/\beta}$
& $\dfrac{\beta\big(1-u^{1+1/\beta}\big)}{(1-u)\,u^{1/\beta}\,(\beta+1)}$ \\
Weibull $(\lambda,k)$ \citep{Lawless2003,JohnsonKotz1994}
& $\frac{1}{\lambda}\{-\log(1-u)\}^{1/k}$
& $\dfrac{e^{t_u}}{t_u^{1/k}}\,\Gamma\!\big(1+\tfrac{1}{k},\,t_u\big),\ \ t_u=-\log(1-u)$ \\
Gamma $(\theta,k)$ \citep{Lawless2003,JohnsonKotz1994}
& $Q(u)=\theta\,G^{-1}(u;k)\footnote{$G^{-1}(u;k)$ denotes the inverse CDF (quantile function) of the Gamma distribution with shape parameter $k$ and unit scale.}$ 
& $\displaystyle \frac{\Gamma\!\big(k+1,\,t_u\big)}{(1-u)\,t_u},\ \ t_u={Q(u)}\footnote{$\Gamma(a,x)$ denotes the upper incomplete gamma function $\int_x^\infty t^{a-1}e^{-t}\,dt$.}$ 
\\
LMRQD \citep{Midhu2013}
& $-(\alpha+\mu)\log(1-u) - 2\alpha u$
& $1 + \dfrac{\mu + \alpha u}{-(\alpha + \mu)\log(1-u) - 2\alpha u}$ \\
\bottomrule
\end{tabularx}
\end{table}

\begin{figure}[ht]
\centering
\includegraphics[width=0.99\linewidth]{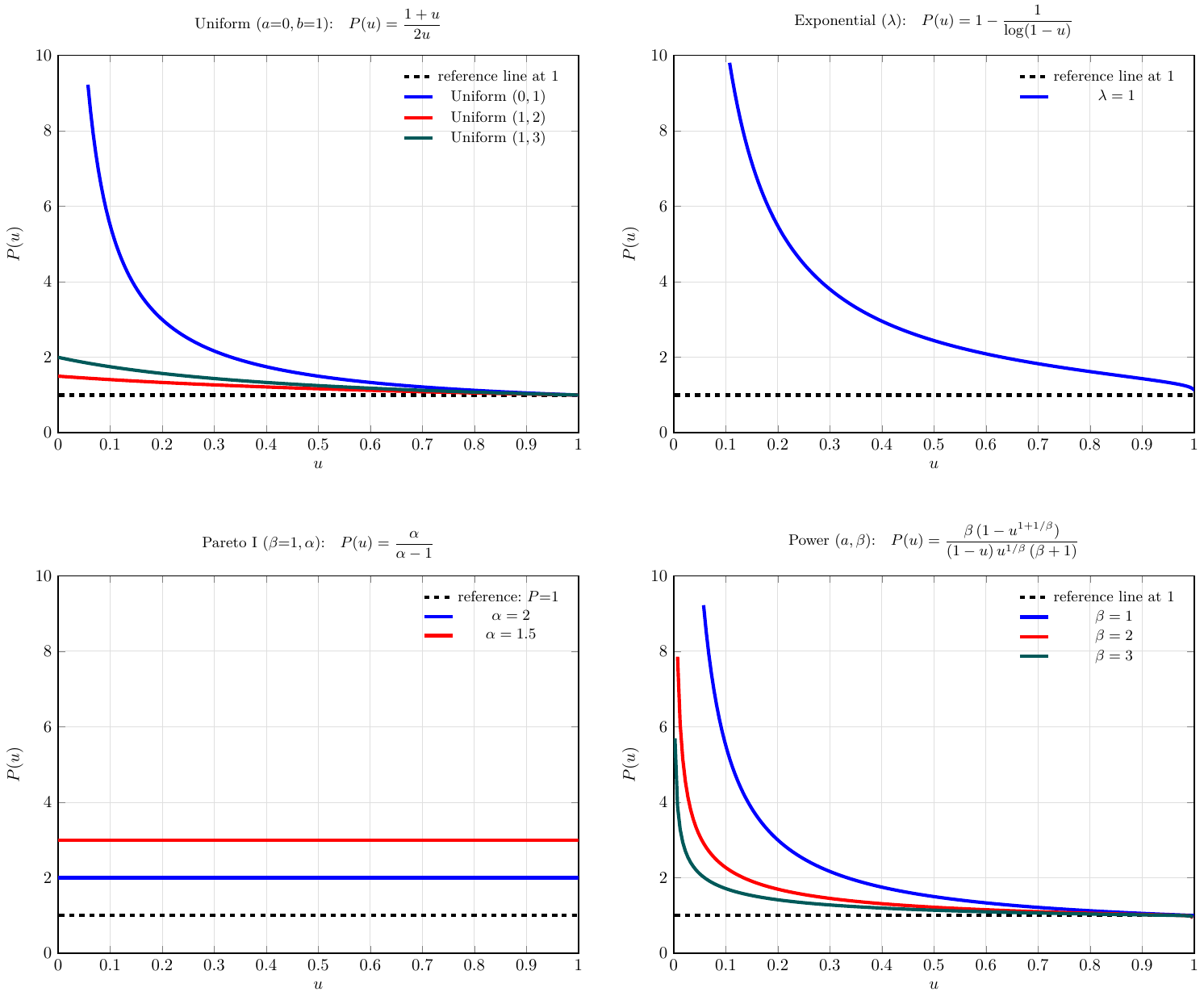}
\caption{Plot of $\QEPF$ for common lifetime distributions.} 
\label{fig:qepf-panels}
\end{figure}


\subsection{Relationship to classical quantile-based measures}
\label{rem:relationship}

In comparative effectiveness research, and especially in the evaluation of biosimilar products, regulatory frameworks emphasize demonstrating clinical similarity to a reference therapy while preserving efficacy and safety in relevant patient subgroups \citep{FDA_Biosim_2015, EMA_Biosim_2014, Chow2013}. Because very favorable outcomes (e.g., longest survival, fastest recovery, largest biomarker reduction) are often concentrated in the upper tail, an explicitly tail-focused metric complements mean, median, or fixed-threshold summaries. 

The quantile-based effectiveness persistence function, $\QEPF$, which measures the average multiplicative uplift among responders above the entry threshold $Q(u)$, thereby quantify how strongly large benefits persist in the top $(1-u)$ fraction.

Beyond its clinical interpretability, $\QEPF$ admits two canonical identities that anchor it within well-known quantile frameworks. 
First, the \emph{connection via the Lorenz curve} relates $\QEPF$ to the upper-tail share of aggregate benefit, tying persistence directly to inequality concentration \citep{Gastwirth1971}. 
Second, the \emph{connection via the Total Time on Test (TTT) transform} expresses $\QEPF$ on the quantile scale in terms of accumulated durability beyond percentile $u$, making tail persistence operational through reliability concepts \citep{BergmanKlefsjo1984,Aarset1987,NairSankaranKumar2008,Nair2013}. 
Together, these links place $\QEPF$ within familiar inequality and reliability toolkits while preserving a tail-local, scale-free interpretation suited to biosimilar evaluations \citep{KleinMoeschberger2003}.




\begin{proposition}[Expression via the Lorenz curve]\label{prop:lorenz}
Let $X$ be a nonnegative outcome with quantile function $Q(u)$ and finite mean
$\mu=\int_0^1 Q(p)\,dp$. Define the Lorenz curve
\[
L(u)=\frac{1}{\mu}\int_0^u Q(p)\,dp,\qquad u\in[0,1]
\]
\citep{Gastwirth1971}. Then
\[
\QEPF \;=\; \frac{\mu\{1-L(u)\}}{(1-u)\,Q(u)}\,.
\]
\end{proposition}

\noindent
Here $1-L(u)$ is the upper-tail share of aggregate benefit, so the ratio expresses tail persistence as an inequality-adjusted average uplift, which is useful for reading concentration of benefit among high responders in biosimilar evaluations.


\begin{proposition}[Expression via Total Time on Test]\label{prop:ttt}
On the quantile scale, define the Total Time on Test (TTT) transform
\[
T(u)=\int_0^u (1-p)\,Q'(p)\,dp,\qquad T(1)=\int_0^1 Q(p)\,dp = \mu
\]
\citep{BergmanKlefsjo1984,Aarset1987,Nair2013}. Then
\[
\QEPF \;=\; 1 + \frac{T(1)-T(u)}{(1-u)\,Q(u)}\,.
\]
\end{proposition}

\noindent
$T(1)-T(u)$ accumulates durability beyond percentile $u$, and normalization by $(1-u)Q(u)$ scales by tail mass and the entry threshold; differences in $T$ between arms diagnose preservation of durability among top responders \citep{KleinMoeschberger2003}. Higher-order TTT extensions appear in \citep{NairSankaranKumar2008}.



\subsection{L-moments of persistence}
\label{subsec:lmom_persistence}

L-moments are expectations of linear combinations of order statistics with superior small-sample stability and outlier robustness \citep{Hosking1990,HoskingRC1990}.

Let $X_{1:r} \le \cdots \le X_{r:r}$ be the order statistics from an i.i.d.\ sample of size $r$ from the distribution of $X$. The $r^{th}$ L-moment is defined as \citep{Hosking1990,HoskingRC1990}
\[
\lambda_r \;=\; \frac{1}{r}\sum_{k=0}^{r-1}(-1)^k \binom{r-1}{k}\, \mathbb{E}\!\left[X_{\,r-k:r}\right],
\qquad r=1,2,3,\ldots
\]

Let $Q(u)=F^{-1}(u)$ denote the quantile function of $X$. The $r^{th}$ L-moment in quantile terms can be expressed as \citep{Nair2013}
\[
\lambda_r \;=\; \int_{0}^{1} Q(u)\,
\bigg\{\sum_{k=0}^{r-1} (-1)^{\,r-1-k}
\binom{r-1}{k}\binom{r-1+k}{k}\, u^{k}\bigg\}\,du,
\qquad r=1,2,3,\ldots
\]

L-moments are defined for any random variable with a finite mean, and finiteness of $\lambda_1$ implies finiteness of all $\lambda_r$ ($r\ge1$) \citep{Hosking1990,HoskingRC1990}. They generally have lower sampling variances. Like conventional moments, L-moments can be used as summary measures of probability distributions (and samples), to identify distributions, and to fit models to data \citep{Nair2013}.

To connect L-moments with persistence in the upper tail, fix a threshold $t>0$ and define the \emph{scaled tail} (normalized excess) random variable
\[
Y \;=\; \frac{X}{t}\,\bigg|\, (X>t),
\]
which represents the behavior of $X$ beyond $t$ in units of $t$. This construction is related to the conditional excess framework in extreme value theory \citep{BalkemaDeHaan1974,Embrechts1997}. In the following theorem, we show that the first L-moment of $Y$ (the excess random variable) equals $\QEPF$.


\begin{theorem}[First L-moment of normalized excess]
\label{thm:firstlmom}
For $t = Q(u)$ with $u \in (0,1)$, the first L-moment of $Y$ is
\[
\lambda_1^*(u) = \frac{1}{(1-u)Q(u)} \int_u^1 Q(p)\, dp =: \QEPF.
\]
\end{theorem}

\begin{proof}
First, compute the distribution of $Y$:
\[
P(Y \le y) = P\!\left( \frac{X}{t} \le y \,\bigg|\, X > t \right)
= \frac{P(t < X \le yt)}{P(X > t)}.
\]
Since $P(X > t) = 1 - F(t)$ and $P(X \le s) = F(s)$, we have
\[
P(Y \le y) = \frac{F(yt) - F(t)}{1 - F(t)}.
\]

Now set $t = Q(u)$ so that $F(t) = u$. Then
\[
P(Y \le y) = \frac{F(y Q(u)) - u}{1 - u}.
\]
Let $v = P(Y \le y)$, then
\[
F(y Q(u)) = u + (1-u)v.
\]
Applying the quantile function $Q(\cdot)$ gives
\[
y Q(u) = Q(u + (1-u)v),
\]
so the quantile function of $Y$ is
\[
Q_Y(v) = \frac{Q(u + (1-u)v)}{Q(u)}, \qquad v \in [0,1].
\]

The first L-moment is
\[
\lambda_1^*(u) = \int_0^1 Q_Y(v)\, dv = \int_0^1 \frac{Q(u+(1-u)v)}{Q(u)}\, dv.
\]
Substitute $p = u+(1-u)v$, so $dp = (1-u)\, dv$, giving
\[
\lambda_1^*(u) = \frac{1}{(1-u)Q(u)} \int_u^1 Q(p)\, dp.
\]
Hence, $\lambda_1^*(u) = \QEPF$.
\end{proof}

\begin{remark}
The identity $\lambda_1^*(u) = \QEPF$ provides a rigorous interpretation of persistence as the mean of the scaled tail beyond the threshold $Q(u)$, consistent with the L-moment framework \citep{Hosking1990}. Because L-moments exist whenever the mean exists, the presence of $\lambda_1^*(u)$ guarantees the existence of all higher L-moments for the normalized tail \citep{Hosking1990,HoskingRC1990}.
\end{remark}


\begin{remark}[Biosimilar tail persistence]
\(\QEPF\) measures the average uplift beyond the entry threshold \(Q(u)\) among the top \((1-u)\) responders. This focuses the analysis on durability where it matters most clinically. Because \(\QEPF\) is \emph{scale-invariant}, comparisons across different units or transformations remain transparent; because it is \emph{location-sensitive}, baseline shifts (e.g., due to supportive care or concomitant therapy) show up as reduced persistence, which is informative when assessing real-world comparability.

\smallskip
\noindent The identity $\QEPF=1+\MRQ{u}/Q(u)$ separates the \emph{residual benefit} ($\MRQ{u}$, the average additional benefit beyond $Q(u)$ from the entry level $Q(u)$. If persistence differs, this split shows whether the gap is due to less remaining gain among top responders or to a higher threshold they must exceed. The Lorenz and TTT forms give complementary views: the Lorenz link shows how much of the total benefit is concentrated in the upper tail; the TTT link shows how durability accumulates above \(Q(u)\). Together, these summaries make tail differences easy to read and report, adding a clear, patient-focused perspective to standard mean or median endpoints.


\end{remark}
\section{Stationary Points of \texorpdfstring{$\QEPF$}{\QEPF}}
\label{sec:beta}
\noindent
Although $\QEPF$ usually declines toward $1$ as $u\uparrow 1$, in some settings the curve can flatten or turn, creating local points where it stops rising or falling. These \emph{stationary points} mark a shift in how much strong benefit continues among the best responders. Finding and interpreting them helps show whether a biosimilar maintains the reference therapy's pattern of durable benefit in the high-response tail, complementing mean or median summaries.

\begin{theorem}\label{thm:stationarity}
If the derivative of $\QEPF$, $\QEP'(u^*)=0$ for some $u^*\in(0,1)$, then stationarity of $\QEPF$ at a quantile $u^*$ leads to a simple balance:
\begin{equation}\label{eq:H=1/M+1/Q}
\QEP'(u^*)=0 \;\Longrightarrow\; H(u^*)=\frac{1}{M(u^*)}+\frac{1}{Q(u^*)},
\end{equation}
where $H(u)=\{(1-u)Q'(u)\}^{-1}$ is the hazard-quantile and $M(u)$ is the mean residual quantile. Equivalently,
\[
\QEP(u^*)=H(u^*)\,M(u^*).
\]
\end{theorem}

At such a point, persistence equals the product of the instantaneous risk (on the quantile scale) and the expected remaining gain beyond the threshold. This balance offers a clear diagnostic for changes in tail persistence, useful both statistically and clinically.

\begin{proof}\leavevmode
From \eqref{eq:pID1}, we have $\QEP(u)=1+M(u)/Q(u)$.
By differentiating, we get:
\[
\QEP'(u)=\frac{M'(u)Q(u)-Q'(u)M(u)}{Q(u)^2}.
\]
At $u^*$, stationarity $\QEP'(u^*)=0$ gives
$$M'(u^*)=\frac{M(u^*)\,Q'(u^*)}{Q(u^*)}$$. Substituting $Q'(u)=\frac{1}{(1-u)H(u)}$
\begin{equation}\label{eq:stationarity-core}
M'(u^*)=\frac{M(u^*)\,\frac{1}{(1-u^*)H(u^*)}}{Q(u^*)}
\end{equation}
But from \citep{Nair2013}, we have 
\begin{equation}\label{eq:MDash}
M'(u) = \frac{M(u) - H(u)^{-1}}{1-u}.
\end{equation}

Substituting \eqref{eq:MDash} into \eqref{eq:stationarity-core} and re-arranging terms yields

\[
H(u^*) = \frac{1}{M(u^*)} + \frac{1}{Q(u^*)}
\]

\end{proof}

\begin{remark}
At a $\QEPF$ extremum, the hazard-quantile equals the sum of reciprocals of past and future survival scales. This identity reflects a dynamic balance: the instantaneous risk is exactly offset by the combined contributions of accumulated and remaining life. Clinically, it marks a transition from risk-driven to survival-driven prognosis-a stabilization phase often associated with durable therapies and long-term responders.
\end{remark}

\begin{remark}
Equivalently, using $\QEP(u)=1+\dfrac{M(u)}{Q(u)}$, \eqref{eq:H=1/M+1/Q} can be expressed as
\[
\QEP(u^*) = H(u^*)\,M(u^*).
\]
This identity reveals a fundamental balance: at the turning point, the $\QEP(u^*)$ equals the product of the hazard-quantile and the mean-residual span. It unites two perspectives-instantaneous risk (through $H(u^*)$) and accumulated expectation (through $M(u^*)$-into a single, dimensionally consistent relationship.
\end{remark}

\subsection{Example: \texorpdfstring{Beta$(\alpha=\tfrac{1}{2},\beta=1)$}{Beta(alpha=1/2,beta=1)}}
\label{subsec:beta-example}

\noindent
To illustrate the stationarity points, we use a simple Beta$(\tfrac{1}{2},1)$ model because it gives closed-form formulas for the quantile $Q(u)$, the hazard-quantile $H(u)$, and the mean residual quantile $M(u)$. From these, we form $\QEPF$ and see how its shape behaves across upper quantiles. Finally, we use the balance identity at a chosen percentile $u^*$ to solve for a small location shift $a$ that creates a turning point. In a comparative setting (including biosimilar evaluations), this helps check whether a new product keeps the same pattern of durable benefit among top responders.

Consider $X\sim\mathrm{Beta}(\alpha,\beta)$ on $[0,1]$ with $\alpha=\tfrac{1}{2}$ and $\beta=1$. The quantile function and its derivative are
\[
Q(u)=u^2,\qquad q(u)=Q'(u)=2u,\qquad u\in(0,1).
\]

The hazard-quantile is
\[
H(u)=\frac{1}{(1-u)\,q(u)}=\frac{1}{2u(1-u)}.
\]

The mean residual quantile is
\[
M(u)=\frac{1}{1-u}\int_u^1\bigl(Q(p)-Q(u)\bigr)\,dp
=\frac{1+u-2u^2}{3}.
\]

Thus,
\[
\QEP(u)=1+\frac{M(u)}{Q(u)}=\frac{1+u+u^2}{3u^2},\qquad u\in(0,1).
\]

For this Beta quantile, $\QEPF$ is strictly decreasing in $u$, diverging as $u\to 0^+$ and approaching $1$ as $u\to 1^-$. This reflects strong persistence of large responses in the upper tail.
To explore shape changes, consider a location shift:
\[
Q_a(u)=u^2+a,\qquad a>0.
\]
The shifted $\QEPF$ becomes
\[
\QEP_a(u)=1+\frac{M(u)}{Q(u)+a}=1+\frac{\dfrac{1+u-2u^2}{3}}{u^2+a}.
\]

Using the identity
\[
H(u^*)=\frac{1}{M(u^*)}+\frac{1}{Q(u^*)+a},
\]

Solving for a
\[
a=\frac{1}{\,H(u^*)-\dfrac{1}{M(u^*)}\,}-Q(u^*),
\]
the required shifts are:
\[
u^*=0.10:\ a\approx0.35,\qquad
u^*=0.15:\ a\approx0.806,\qquad
u^*=0.20:\ a\approx2.20.
\]

Plot of $\QEPF$ and shifted $\QEP_a(u)$ for Beta$(\tfrac{1}{2},1)$ is given in Figure \ref{fig:qepf-beta-placement-3pts}. At very low percentiles ($u^*=0.10$), only a small positive shift ($a\approx0.35$) is needed to create a turning point. At $u^*=0.15$, the required shift increases to about $a\approx0.806$, reflecting a stronger adjustment as hazard and residual balance begins to change. By $u^*=0.20$, the shift becomes much larger ($a\approx2.20$), because the hazard term declines and the residual span narrows rapidly. Larger $a$ flattens the curve, reducing persistence while preserving hazard and residual structure. This progression illustrates how location shifts can tune $\QEPF$ shape at clinically meaningful percentiles without altering underlying reliability measures.

\begin{figure}[ht]
\centering
\begin{tikzpicture}
\begin{axis}[
  width=0.9\linewidth,
  height=0.65\linewidth,
  xlabel={$u$},
  ylabel={$\QEPF$ and $\QEP_a(u)$},
  xmin=0, xmax=1,
  ymin=0.9, ymax=3.5,
  legend style={
    at={(0.6,0.8)},
    anchor=west,
    font=\small,
    draw=none,
    fill=none,
  },
  grid=both,
  major grid style={gray!20},
  minor grid style={gray!10},
  thick,
  restrict y to domain=0.9:4.0,
  unbounded coords=discard,
]
  \addplot[blue,domain=0.01:1,samples=400]
    {(1 + x + x^2)/(3 * x^2)};
  \addlegendentry{Baseline $P(u)$}

  \addplot[red,domain=0:1,samples=400]
    {(1 + ((1 + x - 2*x^2)/3)/(x^2 + 0.35))};
  \addlegendentry{$\QEP_a(u)$, $a=0.35$}

  \addplot[teal!70!black,domain=0:1,samples=400]
    {(1 + ((1 + x - 2*x^2)/3)/(x^2 + 0.806))};
  \addlegendentry{$\QEP_a(u)$, $a=0.806$}

  \addplot[orange,domain=0:1,samples=400]
    {(1 + ((1 + x - 2*x^2)/3)/(x^2 + 2.20))};
  \addlegendentry{$\QEP_a(u)$, $a=2.20$}

  \addplot[black, dashed, domain=0:1] {1};
  \addlegendentry{$\QEPF=1$ reference}

  \addplot[mark=*,only marks,black]
    coordinates {(0.10, {1 + ((1 + 0.10 - 2*0.01)/3)/(0.01 + 0.35)})};
  \node[anchor=west] at (axis cs:0.10, {1 + ((1 + 0.10 - 2*0.01)/3)/(0.01 + 0.35)}) {\small $u^*=0.10$};

  \addplot[mark=*,only marks,black]
    coordinates {(0.15, {1 + ((1 + 0.15 - 2*0.0225)/3)/(0.0225 + 0.806)})};
  \node[anchor=west] at (axis cs:0.15, {1 + ((1 + 0.15 - 2*0.0225)/3)/(0.0225 + 0.806)}) {\small $u^*=0.15$};

  \addplot[mark=*,only marks,black]
    coordinates {(0.20, {1 + ((1 + 0.20 - 2*0.04)/3)/(0.04 + 2.20)})};
  \node[anchor=west] at (axis cs:0.20, {1 + ((1 + 0.20 - 2*0.04)/3)/(0.04 + 2.20)}) {\small $u^*=0.20$};

\end{axis}
\end{tikzpicture}
\caption{Plot of $\QEPF$ and shifted $\QEP_a(u)$ for Beta$(\tfrac{1}{2},1)$}
\label{fig:qepf-beta-placement-3pts}
\end{figure}

\section{Non-parametric Estimation of \texorpdfstring{$\QEPF$}{QEPF}}
\label{sec:empirical-qepf}

Let $X_1,\dots,X_n$ be a random sample from a distribution with quantile function $Q(u)$.
Denote $X_{(1)}\le\cdots\le X_{(n)}$ be the order statistics. An empirical estimator of $Q(u)$ is (see \citealp{Parzen1979})
\[
Q_n(u)\;=\;X_{(j)} \quad \text{for } \frac{j-1}{n}<u\le\frac{j}{n},\;\; j=1,\dots,n,
\]
equivalently $Q_n(u)=X_{(k)}$ with $k=\lceil nu\rceil$.

For $V(u)=\frac{1}{1-u}\int_u^1 Q(t)\,dt$, its empirical (Riemann-sum) version is
\begin{equation}
\label{eq:Vhat}
V_n(u)
\;\equiv\;
\frac{\displaystyle \frac{1}{n}\sum_{i=k+1}^n Q_n\!\left(\frac{i}{n}\right)}{\,1-\frac{k}{n}\,}
\;=\;
\frac{1}{n-k}\sum_{i=k+1}^n X_{(i)},
\qquad k=\lceil nu\rceil,\;\; k<n.
\end{equation}
Thus \eqref{eq:Vhat} is the right Riemann sum of $Q_n$ over $[k/n,1]$, normalized by the tail length.

Using \eqref{eq:defQEPF}, the empirical estimator for $\QEPF$ can be obtained as

\begin{equation}
\label{eq:QEPFhat}
\widehat{\QEP}_n(u)\;=\;P_n(u)
\;=\;
\frac{ \displaystyle \frac{1}{n-k}\sum_{i=k+1}^{n} X_{(i)} }{ X_{(k)} },
\qquad k<n.
\end{equation}

%

As $n\to\infty$ with $k/n\to u$,
\[
X_{(k)} \xrightarrow{a.s.} Q(u),
\qquad
\frac{1}{n-k}\sum_{i=k+1}^{n} X_{(i)} \xrightarrow{a.s.} \Vitality(u),
\]
so $\widehat{\QEP}_n(u)\xrightarrow{a.s.}\QEPF$.
Moreover, using classical quantile-process and $L$-statistic theory together with a delta method
argument,\footnote{See, e.g., \citep{Bahadur1966,Kiefer1967,vdVW1996,Serfling1980}.}
\[
\sqrt{n}\big\{\widehat{\QEP}_n(u)-\QEPF\big\}\ \rightsquigarrow\ \mathcal{N}\!\big(0,\ \sigma^2_{\QEPF}\big).
\]

For small $n$ and $u$ extremely close to 1, $\widehat{\QEP}_n(u)$ can be unstable; we recommend a trimmed upper interval (e.g., $u\le 0.98$ for small $n$). In practice, we use Bootstrap method to compute variability. To assess the finite sample properties of the estimate, we have conducted a simulation study, and is presented in Section \ref{sec:sim1}.

\section{Two-Sample Equivalence Test based on \texorpdfstring{$\QEPF$}{QEPF}}
\label{sec:2sample_eq_test-qepf}

In comparative effectiveness research, particularly in biosimilar evaluation, regulatory frameworks emphasize demonstrating clinical equivalence to a reference product in terms of efficacy and safety \citep{FDA_Biosim_2015, EMA_Biosim_2014, Chow2013}. Conventional summaries (e.g., means or medians of time-to-event) can obscure differences in the upper tail, precisely where durable benefits among top responders are most clinically relevant.

We propose an \emph{upper-tail equivalence} test based on $\QEPF$, which captures persistence of benefit among the top $(1-u)$ fraction relative to the entry threshold $Q(u)$ \citep{Ando2022, AcerbiTasche2002, RockafellarUryasev2002}. Let $Q_{\text{ref}}(u)$ and $Q_{\text{bio}}(u)$ denote the outcome quantile functions for the reference and biosimilar groups, respectively, and let $\mathcal{P}_{\text{ref}}(u)$ and $\mathcal{P}_{\text{bio}}(u)$ denote their corresponding persistence functions. Fix a trimmed upper-tail interval
\[
\mathcal{U} = [u_1, u_2] \subset (0,1), \qquad 0 < u_1 < u_2 < 1,
\]
bounded away from 0 and 1 to avoid instability in extreme tails.

\subsection{Equivalence Hypotheses and Test Statistic}


For $u \in \mathcal{U}$, we test
\begin{equation}
\label{eq:QEPF_EQ_hypotheses}
\begin{aligned}
H_0:\ & {\mathcal{P}}_{\text{ref}}(u) = {\mathcal{P}}_{\text{bio}}(u)\ \text{for all } u \in \mathcal{U},\\
H_1:\ & {\mathcal{P}}_{\text{ref}}(u) \ne {\mathcal{P}}_{\text{bio}}(u)\ \text{for at least one } u \in \mathcal{U}.
\end{aligned}
\end{equation}
  Let $\widehat{\mathcal{P}}_{\text{ref}}(u)$ and $\widehat{\mathcal{P}}_{\text{bio}}(u)$ denote sample-based estimators of $\mathcal{P}_{\text{ref}}(u)$ and $\mathcal{P}_{\text{bio}}(u)$ (see \Cref{sec:empirical-qepf}). A natural test statistic is the scaled supremum of the estimated difference:
\begin{equation}
\label{eq:TestStat}
\widehat{T}_{\text{EQ}} =
\sqrt{\frac{n_{\text{ref}} \, n_{\text{bio}}}{n_{\text{ref}} + n_{\text{bio}}}}
\ \sup_{u \in \mathcal{U}} \big| \widehat{\mathcal{P}}_{\text{ref}}(u) - \widehat{\mathcal{P}}_{\text{bio}}(u) \big|.
\end{equation}

Here $n_{\text{ref}}$ and $n_{\text{bio}}$ denote the sample sizes in the
reference and biosimilar arms, respectively.

\subsection{Bootstrap Calibration and Decision Rule}


We calibrate the critical value by bootstrapping the \emph{supremum process} over $\mathcal{U}$. In practice, the multiplier or paired bootstrap yields valid approximations for suprema of empirical/quantile-derived functionals.
\begin{enumerate}
\item Generate $B$ bootstrap replicates under the pooled null (or via a multiplier bootstrap over $u\in\mathcal{U}$).
\item For each replicate $b=1,\dots,B$, compute
\[
\widehat{T}_{\text{EQ}}^{(b)}
=
\sqrt{\frac{n_{\text{ref}}\,n_{\text{bio}}}{n_{\text{ref}} + n_{\text{bio}}}}
\ \sup_{u \in \mathcal{U}} \left| \widehat{\mathcal{P}}_{\text{ref}}^{(b)}(u) - \widehat{\mathcal{P}}_{\text{bio}}^{(b)}(u) \right|.
\]
\item Let $c_{\alpha}^{\text{boot}}$ be the bootstrap $(1-\alpha)$ critical value for the margin-adjusted supremum.
\end{enumerate}
\emph{Decision rule:} Reject $H_0$ (declare \emph{equivalence} on $\mathcal{U}$) if
\[
\widehat{T}_{\text{EQ}} \le c_{\alpha}^{\text{boot}}.
\]

\noindent
Because $\QEPF$ typically attenuates toward $1$ as $u\to1$, differences between two $\QEPF$ curves are hard to detect in the extreme upper tail. We therefore recommend capping the right endpoint of $\mathcal{U}$ at $u_2 \le 0.90$.

\section{Simulation Study}
\label{sec:simstudy}

In the following sections, we conducted simulation studies to assess the performance of the empirical estimator and the proposed equivalance test.

\subsection{Empirical Estimator \texorpdfstring{$\widehat{\QEP}(u)$}{QEPFhat}}
\label{sec:sim1}

To assess the finite-sample performance (bias and MSE) of the empirical
estimators defined in \Cref{sec:empirical-qepf} by comparing to the true \texorpdfstring{$\QEPF$}{QEPF}, we consider the distributions, $LMRQD(\alpha{=}0.5,\mu{=}5)$, Weibull$(k{=}2,\lambda{=}1)$, and Pareto$(\alpha{=}2.5,\sigma{=}1)$
(light, medium, heavy tails) at $u\in\{0.90,0.95\}$ and sample-size $n\in\{25,50,100,1000\}$.
For each setting we run $1000$ Monte Carlo replications, compute $\widehat{\QEP}(u)$ 
and report bias and MSE of the estimates of $\widehat{\QEP}(u)$.

\begin{figure}[ht]
\centering
\includegraphics[width=0.99\linewidth]{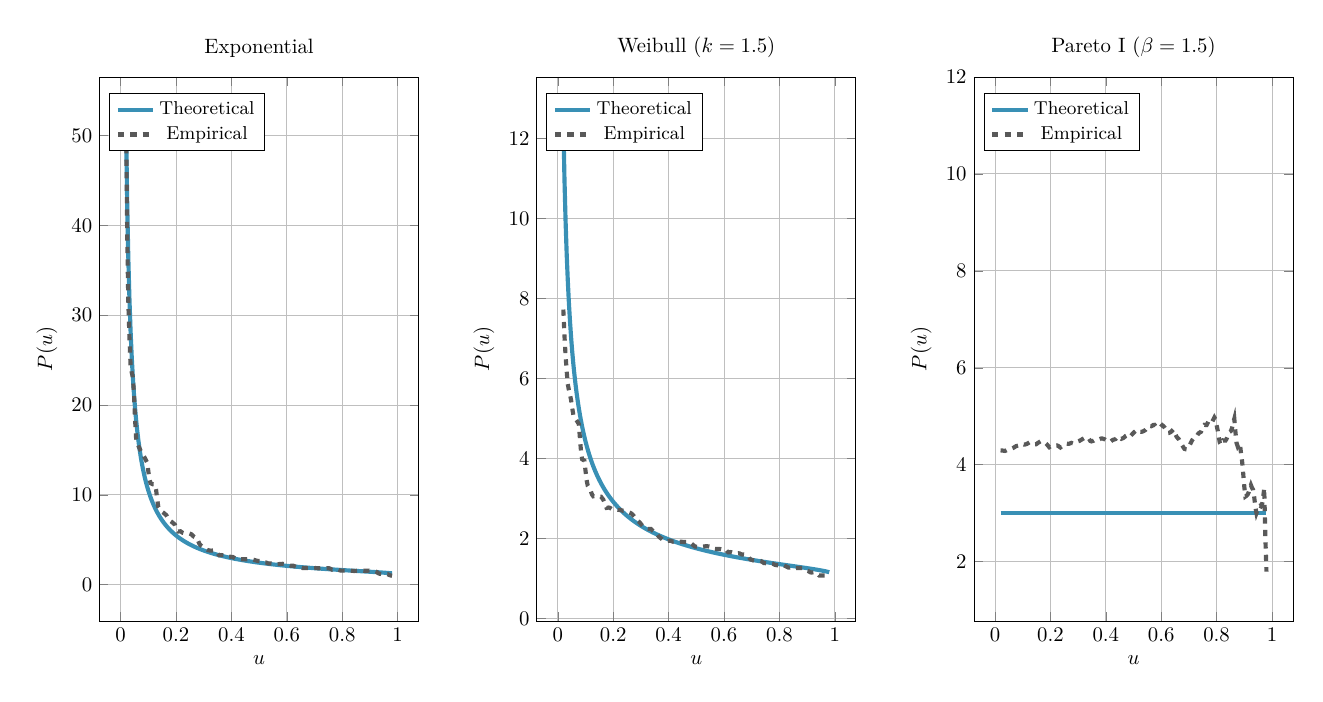}
\caption{Plot of $\widehat{\QEP}(u)$}
\label{fig:qepf-sim-panels}
\end{figure}

From Table~\ref{tab:bias-mse-combined}, for LMRQD and Weibull, bias and MSE decrease rapidly as $n$ increases, with Weibull showing near-zero bias even for moderate $n$. Pareto exhibits much higher variability, especially at $u=0.95$, where MSE remains large for small samples due to heavy-tail behavior. In Figure~\ref{fig:qepf-sim-panels}, we plot $\widehat{\QEP}(u)$ for the mentioned distributions used for simulations.

\begin{table}[ht]
\centering
\caption{Bias and MSE for different distributions at $u\in\{0.85,0.90,0.95\}$ across sample sizes.}
\label{tab:bias-mse-combined}
\begin{tabular}{llrrrr}
\toprule
Distribution & $u$ & $n$ & True & Bias & MSE \\
\midrule
\multirow{12}{*}{LMRQD $(\alpha=0.5,\ \mu=5)$}
& 0.85 & 25   & 1.5660 &  0.0050 & 0.1353 \\
&      & 50   & 1.5660 &  0.0157 & 0.0595 \\
&      & 100  & 1.5660 &  0.0143 & 0.0272 \\
&      & 1000 & 1.5660 &  0.0037 & 0.0027 \\
& 0.90 & 25   & 1.4633 &  0.0494 & 0.1555 \\
&      & 50   & 1.4633 &  0.0379 & 0.0639 \\
&      & 100  & 1.4633 &  0.0085 & 0.0268 \\
&      & 1000 & 1.4633 & -0.0003 & 0.0025 \\
& 0.95 & 25   & 1.3526 &  0.0334 & 0.1519 \\
&      & 50   & 1.3526 &  0.0129 & 0.0696 \\
&      & 100  & 1.3526 &  0.0206 & 0.0296 \\
&      & 1000 & 1.3526 &  0.0036 & 0.0031 \\
\midrule
\multirow{12}{*}{Weibull $(k=2,\ \lambda=1)$}
& 0.85 & 25   & 1.2206 & -0.0002 & 0.0150 \\
&      & 50   & 1.2206 &  0.0035 & 0.0067 \\
&      & 100  & 1.2206 &  0.0082 & 0.0033 \\
&      & 1000 & 1.2206 &  0.0001 & 0.0003 \\
& 0.90 & 25   & 1.1862 &  0.0066 & 0.0186 \\
&      & 50   & 1.1862 &  0.0149 & 0.0079 \\
&      & 100  & 1.1862 &  0.0063 & 0.0034 \\
&      & 1000 & 1.1862 &  0.0003 & 0.0003 \\
& 0.95 & 25   & 1.1472 &  0.0146 & 0.0218 \\
&      & 50   & 1.1472 &  0.0133 & 0.0121 \\
&      & 100  & 1.1472 &  0.0063 & 0.0039 \\
&      & 1000 & 1.1472 &  0.0000 & 0.0004 \\
\midrule
\multirow{12}{*}{Pareto I $(\alpha=2.5,\ \sigma=1)$}
& 0.85 & 25   & 1.6667 & -0.0185 & 0.5188 \\
&      & 50   & 1.6667 & -0.0163 & 0.1961 \\
&      & 100  & 1.6667 & -0.0012 & 0.1178 \\
&      & 1000 & 1.6667 &  0.0022 & 0.0133 \\
& 0.90 & 25   & 1.6667 &  0.0471 & 0.9631 \\
&      & 50   & 1.6667 &  0.0052 & 0.2645 \\
&      & 100  & 1.6667 &  0.0017 & 0.1673 \\
&      & 1000 & 1.6667 &  0.0010 & 0.0189 \\
& 0.95 & 25   & 1.6667 &  0.0367 & 1.4209 \\
&      & 50   & 1.6667 &  0.0918 & 2.0180 \\
&      & 100  & 1.6667 & -0.0281 & 0.2357 \\
&      & 1000 & 1.6667 &  0.0089 & 0.0413 \\
\bottomrule
\end{tabular}
\end{table}




\subsection{Two-Sample Equivalence Test}
\label{sec:sim2}

In this section, we have conducted a simulation study to evaluate the finite-sample performance of the test statistic on a trimmed upper-tail interval $\mathcal{U} = [0.60, 0.90]$. Two independent samples (reference vs biosimilar) were generated under six scenarios: three non-equality cases (\texttt{Gamma vs LMRQD}, \texttt{Gamma vs Lognormal}, \texttt{LMRQD vs Lognormal}) and three equality cases (\texttt{Gamma vs Gamma}, \texttt{LMRQD vs LMRQD}, \texttt{Lognormal vs Lognormal}), with sample sizes per group $n \in \{50, 100, 300, 500\}$. For each replicate, the test statistic was computed. A total of 2000 Monte Carlo replicates were performed for each design point.

In Table~\ref{tab:simBootZ}, results indicate that under non-equality scenarios, power increased substantially with sample size: for example, \texttt{LMRQD vs Lognormal} reached $\approx$ $99.6\%$ at $n = 500$, while \texttt{Gamma vs Lognormal} achieved $\approx 98.2\%$ at the same size. However, under the alternative and smaller samples ($n = 50$), power was modest for some contrasts (e.g., $39.8\%$ for \texttt{LMRQD vs Lognormal} and $64.9\%$ for \texttt{Gamma vs LMRQD}), which we attribute to the tail restriction imposed by $\mathcal{U}$; with fewer observations in the upper quantile range, the supremum statistic becomes less sensitive to differences in persistence.

\begin{table}[ht]
\centering
\scriptsize
\caption{Empirical Power and Size of Bootstrap-Calibrated $\QEPF$ Test ($\mathcal{U}=[0.60,0.90]$, $\alpha=0.05$)}
\label{tab:simBootZ}
\begin{tabular}{lcccc}
\toprule
\textbf{Scenario} & \textbf{$n=50$} & \textbf{$n=100$} & \textbf{$n=300$} & \textbf{$n=500$} \\
\midrule
\multicolumn{5}{l}{\textbf{Non-Equality}} \\
\midrule
Gamma vs LMRQD        & 64.95 & 94.45 & 100.0 & 100.0 \\
Gamma vs Lognormal     & 81.35 & 90.85 & 96.85 & 98.20 \\
LMRQD vs Lognormal     & 39.75 & 71.70 & 98.30 & 99.55 \\
\midrule
\multicolumn{5}{l}{\textbf{Equality}} \\
\midrule
Gamma vs Gamma         & 1.60  & 2.10  & 2.40  & 3.20 \\
LMRQD vs LMRQD         & 1.95  & 2.95  & 3.50  & 2.60 \\
Lognormal vs Lognormal & 1.50  & 1.95  & 2.60  & 2.40 \\
\bottomrule
\end{tabular}
\end{table}

Under equality scenarios, the empirical size was conservative relative to $\alpha = 0.05$, ranging from approximately $1.5\%$ to $3.5\%$, reflecting robustness against false equivalence but suggesting slight over-calibration in finite samples. Overall, the bootstrap-calibrated test demonstrates strong sensitivity to tail differences and reliable type I error control, with performance improving as sample size and bootstrap replicates increase, supporting its practical utility in biosimilar equivalence assessment \citep{FDA_Biosim_2015, EMA_Biosim_2014, Chow2013}.

\section{Real-Data Application: SMARTER Trial}
\label{sec:realdata}

To illustrate the practical utility of the proposed two-sample test based on the quantile-based persistence function, we apply the method to a real-dataset from SMARTER trial. The trial evaluated a village doctor-led mobile health intervention for cardiovascular risk reduction (Intervention group) compared to usual care (Control group) in systolic blood pressure (SBP) \citep{SMARTER_BMJ2025, SMARTER_NCT05645640, Li_2025}. 


\begin{figure}[ht]
\centering
\includegraphics[width=0.9\linewidth, height=0.4\textheight]{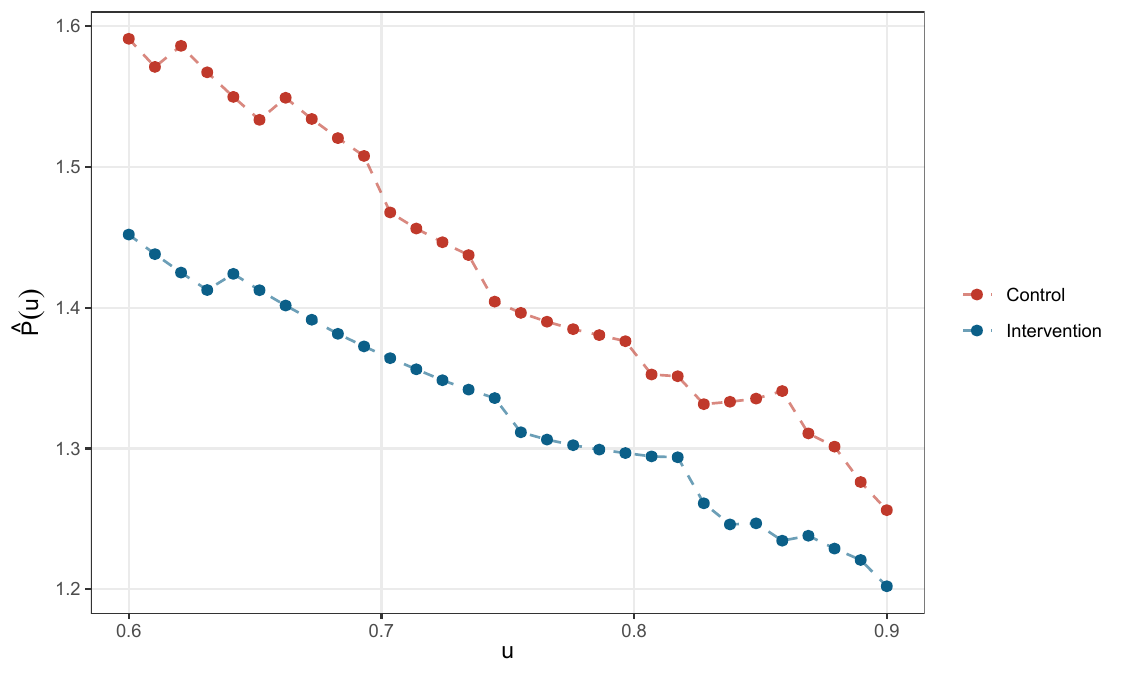}
\caption{$\widehat{\QEP}(u)$ curves for SBP reduction at 12 months over $\mathcal{U}=[0.60,0.90]$.}
\label{fig:smarter_curves}
\end{figure}

We computed $\widehat{\QEP}(u)$ for each group over $\mathcal{U}=[0.60,0.90]$. Figure~\ref{fig:smarter_curves} shows the estimated curves for SBP reduction at 12 months. The intervention group exhibits consistently higher persistence across the trimmed upper tail. Applying the proposed test on $\mathcal{U}=[0.60,0.90]$ yielded $\widehat{T}_{\text{EQ}}=5.03$ with bootstrap p-value $0.0001$, strongly rejecting equality.

\begin{table}[ht]
\centering
\scriptsize
\caption{Sensitivity analysis for $\widehat{\QEP}(u)$ under different upper-tail intervals.}
\label{tab:smarter_sensitivity}
\begin{tabular}{ccccccc}
\toprule
$U_{\text{lower}}$ & $U_{\text{upper}}$ & $\widehat{T}_{\text{EQ}}$ & $n_{\text{eff}}$ & P-value & crit95$_{\text{boot}}$ & Width \\
\midrule
0.10 & 0.90  & 52.00 & 974 & 0.0001 & 29.4 & 0.8 \\
0.15 & 0.90  & 32.30 & 974 & 0      & 16.9 & 0.75 \\
0.20 & 0.90  & 20.90 & 974 & 0.0004 & 10.8 & 0.7 \\
0.25 & 0.90  & 17.20 & 974 & 0      & 8.55 & 0.65 \\
0.30 & 0.90  & 12.80 & 974 & 0      & 6.08 & 0.6 \\
0.35 & 0.90  & 10.30 & 974 & 0      & 5.26 & 0.55 \\
0.40 & 0.90  & 9.66  & 974 & 0      & 4.46 & 0.5 \\
0.45 & 0.90  & 7.68  & 974 & 0      & 4.04 & 0.45 \\
0.50 & 0.90  & 7.36  & 974 & 0      & 3.36 & 0.4 \\
0.55 & 0.90  & 5.18  & 974 & 0.0001 & 2.96 & 0.35 \\
\textbf{0.60} & \textbf{0.90}  & \textbf{5.03}  & \textbf{974} & \textbf{0.0001} & \textbf{2.86} & \textbf{0.3} \\
0.65 & 0.90  & 4.56  & 974 & 0.0002 & 2.60 & 0.25 \\
0.70 & 0.90  & 4.04  & 974 & 0.0001 & 2.29 & 0.2 \\
0.75 & 0.90  & 3.36  & 974 & 0.0005 & 2.19 & 0.15 \\
0.80 & 0.90  & 3.32  & 974 & 0.0003 & 2.07 & 0.1 \\
0.85 & 0.975 & 3.35  & 974 & 0.0018 & 2.38 & 0.125 \\
\bottomrule
\end{tabular}
\end{table}

Sensitivity analysis across multiple intervals confirmed robustness (Table~\ref{tab:smarter_sensitivity}): all p-values were $\leq 0.0004$, and $\widehat{T}_{\text{EQ}}$ ranged from 3.32 to 52.0, with larger intervals producing higher test statistics due to greater tail mass. The primary analysis interval $\mathcal{U}=[0.60,0.90]$ (in bold) yielded a highly significant result, supporting durable benefit among high responders in the intervention group.



The original SMARTER trial reported a mean reduction in systolic blood pressure (SBP) of $7.0$ mmHg (95\% CI: $5.7$ to $8.3$) in the intervention group compared to usual care, alongside a significant decrease in predicted 10-year ASCVD risk \citep{SMARTER_BMJ2025}. Our QEPF-based analysis reveals that, in addition to these average effects, the intervention group exhibits a higher concentration of extreme SBP responders, as evidenced by significantly greater upper-tail persistence. This tail-focused perspective complements the primary endpoint and highlights the durability of benefit among high responders.

\section{Discussion and Conclusion}
\label{sec:discussion-conclusion}

The quantile-based effectiveness persistence function $\QEPF$ provides an interpretable, tail-focused lens on clinical performance: it quantifies the average multiplicative uplift among responders exceeding the entry threshold $Q(u)$. As $u\to1$, $\QEPF$ naturally declines toward $1$, reflecting diminishing tail persistence; conversely, constant $\QEPF$ (e.g., Pareto tails) signals self-similar scaling across upper quantiles.

From a comparative effectiveness standpoint, particularly for biosimilar evaluation, an upper-tail equivalence margin can be pre-specified to encode clinically acceptable differences in persistence among high responders. This tail-sensitive assessment complements mean and or median-based non-inferiority by focusing on durability of favorable outcomes among the top portion of the distribution.

Practical guidance for inference includes working on trimmed upper-tail intervals $\mathcal{U}$ away from extreme endpoints (e.g., $[0.60,0.90]$ in small samples) to improve stability, and adopting cluster-level bootstrap calibration when the design is cluster-randomized, thereby preserving intra-cluster dependence.

This work established core properties of $\QEPF$ and formal connections to classical reliability and inequality tools (Lorenz and TTT transforms), and showed that $\QEPF$ equals the first L-moment of the scaled tail beyond $Q(u)$. We analyzed several common lifetime distributions to illustrate behavior across tail shapes, proposed a simple empirical estimator, and developed a two-sample equivalence test on trimmed tail intervals. Simulation studies demonstrated favorable finite-sample performance, and a real-data application illustrated practical utility in a clinical setting.

Extensions include censoring-aware estimators (e.g., KM-based quantile/vitality), covariate-adjusted tail persistence via quantile regression, robust approaches for heavy-tailed data (including trimming and regularization), and formal power analyses for tail-equivalence testing under various distributional regimes. These developments can further integrate $\QEPF$ into regulatory-quality workflows where sustained benefit among high responders is of primary clinical relevance.

\section*{Declarations}

\begin{itemize}
\item \textbf{Funding:} The authors received no financial support for the research, authorship, and/or publication of this article.
\item \textbf{Conflict of interest:} The authors declare that they have no conflict interests.
\item \textbf{Author contributions:} All authors contributed to the study conception and design. Material preparation and analysis were performed by the authors. The first draft of the manuscript was written by the authors, and all authors commented on previous versions of the manuscript. All authors read and approved the final manuscript.
\end{itemize}

\bibliography{qepf_ref}

\end{document}